\documentclass{article}
\usepackage{amsmath}
\usepackage{amsfonts}

\newcommand{\cc}[1]{\overline{#1}}
\newcommand{\abs}[1]{\left|#1\right|}
\newcommand{\norm}[1]{\left\lVert#1\right\rVert}
\newcommand{\defined}{:=}
\newcommand{\normeq}{\mathring{=}}

\renewcommand{\vec}[1]{\mathbf{#1}}

\DeclareMathOperator{\const}{const}

\newtheorem{theorem}{Theorem}
\newenvironment{proof}{\par\noindent\textbf{Proof} }{\par\hfill$\Box$\par}

\begin{document}

\title{On the solutions of generalized discrete Poisson equation}
\author{Roman Werpachowski\\Center for Theoretical Physics, PAS\\Al. Lotnik\'{o}w 32/46, 02-668 Warsaw, Poland}

\maketitle

\begin{abstract}
The set of common numerical and analytical problems is introduced in the form of the \emph{generalized multidimensional discrete Poisson equation}. It is shown that its solutions with square-summable discrete derivatives are unique up to a constant. The proof uses the Fourier transform as the main tool. The necessary condition for the existence of the solution is provided.
\end{abstract}



\section{Introduction}

The motivation for this paper comes from an attempt to construct
(a discrete version of) the quantum field theory interacting with
a non-trivial gravitational field. Such a theory would describe a
quantum mechanical system with infinitely many degrees of freedom,
assigned to the points of an infinite lattice
$\mathbb{Z}^d$~\cite{kijowski1,kijowski2,kijowski3}. The
multidimentional discrete Poisson equation arises as a natural
tool of such a theory. In this paper we present a proof of the
uniqueness (up to an additive constant) of a class of its
solutions. The existence proof will be the subject of further research.

The equation we deal with may be derived  from the variational
principle:
\begin{equation}
\label{eq:deltaW}
\delta W(\cc{f},f) = 0 \ ,
\end{equation}
where
\[
W(\cc{f},f) \defined \sum_{\vec{n} \in \mathbb{Z}^d} \sum_{k=1}^d \sum_{l=1}^d b_{\vec{n},kl} \cc{(f_\vec{n} - f_{\vec{n} - \vec{e}_k})} (f_\vec{n} - f_{\vec{n} - \vec{e}_l}) \ .
\]
Here, $f$ is our unknown function, a complex sequence defined on the lattice
$\mathbb{Z}^d$ ($f: \mathbb{Z}^d \mapsto \mathbb{C}$), whereas
$b: \mathbb{Z}^d \mapsto \mathbb{C}^{d \times d}$ is a sequence of
$d \times d$ positive Hermitian matrices $b_\vec{n}$ whose spectra
$\sigma(b_\vec{n})$ have common bounds,
\begin{equation}
\label{eq:bbnd} \forall_{\vec{n} \in \mathbb{Z}^d} \qquad
\sigma(b_\vec{n}) \in (b_1, b_2] \ , \qquad 0 \le b_1 \le b_2 <
\infty \ .
\end{equation}
We assume that $f$ fulfills the following condition:
\begin{equation}
\label{eq:condition} \forall_{1 \le k \le d} \qquad \sum_{\vec{n}
\in \mathbb{Z}^d} \abs{f_\vec{n} - f_{\vec{n} - \vec{e}_k}}^2 <
\infty \ .
\end{equation}
This makes $W(\cc{f},f)$ finite, since
\[
b_1 \sum_{k=1}^d \abs{f_\vec{n} - f_{\vec{n} - \vec{e}_k}}^2 <
\sum_{k=1}^d \sum_{l=1}^d b_{\vec{n},kl} \cc{(f_\vec{n} -
f_{\vec{n} - \vec{e}_k})} (f_\vec{n} - f_{\vec{n} - \vec{e}_l})
\le b_2 \sum_{k=1}^d \abs{f_\vec{n} - f_{\vec{n} - \vec{e}_k}}^2 \
.
\]
Varying $W(\cc{f},f)$ over $\cc{f_\vec{n}}$, we derive the
following homogeneous equation for $f$:
\[
\sum_{k=1}^d \sum_{l=1}^d \left[ b_{\vec{n} +
\vec{e}_k,kl} ( f_{\vec{n} + \vec{e}_k} - f_{\vec{n} + \vec{e}_k -
\vec{e}_l} ) -  b_{\vec{n},kl} ( f_\vec{n} - f_{\vec{n} -
\vec{e}_l} ) \right] = 0 \ .
\]
In the present paper, we consider a
general, non-homogeneous case
\begin{equation}
\label{eq:Peq} 
\sum_{k=1}^d \sum_{l=1}^d \left[ b_{\vec{n} +
\vec{e}_k,kl} ( f_{\vec{n} + \vec{e}_k} - f_{\vec{n} + \vec{e}_k -
\vec{e}_l} ) -  b_{\vec{n},kl} ( f_\vec{n} - f_{\vec{n} -
\vec{e}_l} ) \right] = g_\vec{n} \ ,
\end{equation}
which we call the \emph{generalized multidimensional discrete
Poisson equation}. The simplest example is provided by the standard
multidimensional discrete Poisson equation, corresponding to
$b_{\vec{n},kl} = \delta_{k,l}$:
\[
\sum_{k=1}^d \left( f_{\vec{n} + \vec{e}_k} + f_{\vec{n} -
\vec{e}_k} - 2 f_\vec{n} \right) = g_\vec{n} \ .
\]
Of course, adding a constant to a solution $f$ of~\eqref{eq:Peq}
we again obtain a solution. Within the class of functions fulfilling~\eqref{eq:condition}, we prove that any two solutions of
equation~\eqref{eq:Peq} are equal up to an additive constant.

Contrary to the case of ordinary differential equations, there is
no general theorem on the existence and uniqueness of solutions of
discrete equations, neither in one nor in many dimensions. Only
partial results exist, see for example~\cite{hankerson,veit}.
In~\cite{veit}, the uniqueness of solutions vanishing at infinity
($\lim_{\norm{\vec{n}} \to \infty} f_\vec{n} = 0$) has been proved
for a wide class of multidimensional discrete equations.
Unfortunately, this is not sufficient for purposes of the quantum
field theory. Our result presented here is valid for solutions
fulfilling a different condition, namely~\eqref{eq:condition}.

Given two solutions of~\eqref{eq:Peq}, $f$ and $f'$, satisfying
condition~\eqref{eq:condition}, their difference $x_\vec{n}
\defined f_\vec{n} - f'_\vec{n}$
fulfills~\eqref{eq:condition} and solves the homogeneous equation
with $g = 0$:
\begin{equation}
\label{eq:Peq1} \sum_{k=1}^d \sum_{l=1}^d \left[ b_{\vec{n} +
\vec{e}_k,kl} ( x_{\vec{n} + \vec{e}_k} - x_{\vec{n} + \vec{e}_k -
\vec{e}_l} ) -  b_{\vec{n},kl} ( x_\vec{n} - x_{\vec{n} -
\vec{e}_l} ) \right] = 0 \ .
\end{equation}

It is, therefore, sufficient to prove that, within the class of
functions fulfilling~\eqref{eq:condition}, any solution of
~\eqref{eq:Peq1} is constant.

For the non-homogeneous equation~\eqref{eq:Peq}, the existence of
a solution depends very much upon the properties  of the
right-hand side $g$ and will be analyzed elsewhere.


\section{Uniqueness theorem}

\begin{theorem}
\label{thm:homconst}
Let $x: \mathbb{Z}^d \mapsto \mathbb{C}$ be a solution of the homogeneous generalized discrete Poisson equation~\eqref{eq:Peq1} in $d$ dimensions. Let us assume that $x$ has the property~\eqref{eq:condition},
\begin{equation}
\label{eq:sumx2}
\forall_{1 \le k \le d} \qquad \sum_{\vec{n} \in \mathbb{Z}^d} \abs{x_\vec{n} - x_{\vec{n} - \vec{e}_k}}^2 < \infty \ .
\end{equation}
Then
\[
x_\vec{n} = \const \ .
\]
\end{theorem}
Observe that the uniqueness within the class of square-summable
functions:
\[
\sum_{\vec{n} \in \mathbb{Z}^d} \abs{x_\vec{n}}^2 < \infty \ ,
\]
follows easily from the following, standard argument. We multiply
both sides of~\eqref{eq:Peq1} by $\cc{x_\vec{n}}$ and sum over
$\vec{n} \in \mathbb{Z}^d$, obtaining
\[
\sum_{\vec{n} \in \mathbb{Z}^d} \cc{x_\vec{n}} \sum_{k=1}^d \sum_{l=1}^d \left[ b_{\vec{n} + \vec{e}_k,kl} ( x_{\vec{n} + \vec{e}_k} - x_{\vec{n} + \vec{e}_k - \vec{e}_l} ) -  b_{\vec{n},kl} ( x_\vec{n} - x_{\vec{n} - \vec{e}_l} ) \right] = 0 \ .
\]
Changing the order of summation in this expression we get:
\[
\sum_{\vec{n} \in \mathbb{Z}^d} \sum_{k=1}^d \sum_{l=1}^d
b_{\vec{n},kl}  \cc{\left( x_\vec{n} - x_{\vec{n} - \vec{e}_k}
\right)} \left( x_\vec{n} - x_{\vec{n} - \vec{e}_l} \right) = 0 \
.
\]
Due to~\eqref{eq:bbnd}, this implies $x_\vec{n} = \const$.
However, we consider solutions which are not necessarily
square-summable, and the above argument does not work.

\begin{proof}
Consider the following auxiliary quantity $v: \mathbb{Z}^d \otimes
[1,d] \mapsto \mathbb{C}$, defined as
\begin{equation}
\label{eq:vdef}
v_{\vec{n},k} \defined x_\vec{n} - x_{\vec{n} - \vec{e}_k} \ .
\end{equation}

From~\eqref{eq:sumx2} and~\eqref{eq:vdef},  we have that for each
$k$, $(v_{\vec{n},k})$ is a square-summable sequence,
\begin{equation}
\label{eq:vsumsq}
\sum_{\vec{n} \in \mathbb{Z}^d} \abs{v_{\vec{n},k}}^2 < \infty \ .
\end{equation}

The fact that $v$ is square-summable allows us to define another auxiliary quantity $\tilde{v}: [1,d] \otimes [-\pi,\pi]^d \mapsto \mathbb{C}$ as the Fourier transform of $v$,
\begin{equation}
\label{eq:vtdef}
\tilde{v}_k(\vec{s}) \defined \frac{1}{(2\pi)^{d/2}} \sum_{\vec{n} \in \mathbb{Z}^d} v_{\vec{n},k} e^{- i \vec{n} \cdot \vec{s}} \ , \qquad \vec{s} \in [-\pi,\pi]^d \ .
\end{equation}
Additionally,~\eqref{eq:vsumsq} leads to $\tilde{v}_k$ being square-integrable,
\[
\int_{[-\pi,\pi]^d} \abs{\tilde{v}_k(\vec{s})}^2 \mathrm{d}^d \vec{s} < \infty \ .
\]
Due to~\eqref{eq:vdef}, we have for each pair $1 \le k_1,k_2 \le d$ and each $\vec{n} \in \mathbb{Z}^d$
\[
v_{\vec{n},k_1} - v_{\vec{n} - \vec{e}_{k_2}, k_1} = v_{\vec{n},k_2} - v_{\vec{n} - \vec{e}_{k_1}, k_2} \ .
\]
The Fourier transform of this equation goes as follows:
\[
\tilde{v}_{k_1}(\vec{s}) \left(1 - e^{-is_{k_2}} \right) = \tilde{v}_{k_2}(\vec{s}) \left(1 - e^{-is_{k_1}} \right) \ .
\]
Therefore, the following equality is valid for $s_{k_2} \neq 0$: $\tilde{v}_{k_1}(\vec{s}) = \tilde{v}_{k_2}(\vec{s})$ $\left(1 - e^{-is_{k_1}} \right)$ $\left(1 - e^{-is_{k_2}} \right)^{-1}$. The set $\{ \vec{s} \in [-\pi,\pi]^d: s_{k_2} = 0 \}$ has measure zero (in the measure $\prod_{j=1}^d \mathrm{d} s_j = \mathrm{d}^d \vec{s}$). We rewrite the equality as
\begin{equation}
\label{eq:trrule}
\tilde{v}_{k_1}(\vec{s}) \normeq \tilde{v}_{k_2}(\vec{s}) \frac{1 - e^{-is_{k_1}}}{1 - e^{-is_{k_2}}} \ ,
\end{equation}
where $\normeq$ means `equal everywhere in $[-\pi,\pi]^d$ except for a set with measure zero'.

It is be convenient for us to introduce another pair of auxiliary quantities. Let us define $y: \mathbb{Z}^d \otimes [1, d]^2 \mapsto \mathbb{C}$ as
\begin{equation}
\label{eq:ydef}
y_{\vec{n},kl} \defined b_{\vec{n},kl} v_{\vec{n},l} \ .
\end{equation}
Due to the bounds~\eqref{eq:bbnd} on $b_\vec{n}$ and the fact that it is a Hermitian matrix, we have $\abs{b_{\vec{n},kl}} \le b_2$. Inserting this into $\sum_{\vec{n} \in \mathbb{Z}^d} \abs{y_{\vec{n},kl}}^2$, we get
\begin{equation}
\label{eq:ysumsq}
\sum_{\vec{n} \in \mathbb{Z}^d} \abs{y_{\vec{n},kl}}^2 \le b_2 \sum_{\vec{n} \in \mathbb{Z}^d} \abs{v_{\vec{n},k}}^2 < \infty
\end{equation}
for each $1 \le k,l \le d$. Because of~\eqref{eq:ysumsq}, we can define another quantity, $\tilde{y}: [1, d]^2 \otimes [-\pi,\pi]^d \mapsto \mathbb{C}$, as the Fourier transform of $y_{\vec{n},kl}$,
\begin{equation}
\label{eq:ytdef}
\tilde{y}_{kl}(\vec{s}) \defined \frac{1}{(2\pi)^{d/2}} \sum_{\vec{n} \in \mathbb{Z}^d} y_{\vec{n},kl} e^{-i \vec{n} \cdot \vec{s}} \ ,
\end{equation}
$\tilde{y}_{kl}$ is a square-integrable function on the domain $[-\pi,\pi]^d$,
\[
\int_{[-\pi,\pi]^d} \abs{\tilde{y}_{kl}(\vec{s})}^2 \mathrm{d}^d \vec{s} < \infty \ .
\]

With the help of~\eqref{eq:vdef} and~\eqref{eq:ydef}, we write~\eqref{eq:Peq1} as
\begin{equation}
\label{eq:Peq2}
\sum_{k=1}^d \sum_{l=1}^d \left( y_{\vec{n},kl} - y_{\vec{n} + \vec{e}_k,kl} \right) = 0 \ .
\end{equation}
Using the definition of $\tilde{y}$, we can calculate the Fourier transform of~\eqref{eq:Peq2},
\[
\sum_{k=1}^d \sum_{l=1}^d \tilde{y}_{kl}(\vec{s}) \left( 1 - e^{is_k} \right) = 0 \ .
\]
Multiplying both sides by $\cc{\tilde{v}_1(\vec{s})} (1 - e^{is_1})^{-1}$ and using~\eqref{eq:trrule}, we obtain
\[
\sum_{k=1}^d \left( \sum_{l=1}^d \tilde{y}_{kl}(\vec{s}) \right) \cc{\tilde{v}_k(\vec{s})} \normeq 0 \ .
\]
Recalling the definition of $\normeq$, we can integrate this formula over $\vec{s}$, obtaining
\[
\sum_{k=1}^d \int_{[-\pi,\pi]^d} \cc{\tilde{v}_k(\vec{s})} \left( \sum_{l=1}^d \tilde{y}_{kl}(\vec{s}) \right) \mathrm{d}^d \vec{s} = 0 \ .
\]
Since the Fourier transform preserves the $L^2$ scalar product, we have from~\eqref{eq:vtdef},~\eqref{eq:ydef} and~\eqref{eq:ytdef}
\begin{equation}
\label{eq:zero1}
\sum_{k=1}^d \sum_{\vec{n} \in \mathbb{Z}^d} \cc{v_{\vec{n},k}} \left( \sum_{l=1}^d b_{\vec{n},kl} v_{\vec{n},l} \right) = 0 \ .
\end{equation}
By~\eqref{eq:bbnd} and~\eqref{eq:vsumsq}, we have for each $1 \le k \le d$
\[
\sum_{\vec{n} \in \mathbb{Z}^d} \abs{\sum_{l=1}^d b_{\vec{n},kl} v_{\vec{n},l}}^2 < \infty \ .
\]
The last result, together with~\eqref{eq:vsumsq} and the Schwartz inequality, ensures the convergence of the series
\[
\sum_{\vec{n} \in \mathbb{Z}^d} \cc{v_{\vec{n},k}} \left( \sum_{l=1}^d b_{\vec{n},kl} v_{\vec{n},l} \right)
\]
for each $1 \le k \le d$. Thus, we may transform~\eqref{eq:zero1} into
\begin{equation}
\label{eq:zero2}
\sum_{\vec{n} \in \mathbb{Z}^d} \sum_{k=1}^d \sum_{l=1}^d b_{\vec{n},kl} \cc{v_{\vec{n},k}}  v_{\vec{n},l}  = 0 \ .
\end{equation}
From~\eqref{eq:bbnd} we know that for each $\vec{n} \in \mathbb{Z}^d$, we have
\[
\sum_{k=1}^d \sum_{l=1}^d b_{\vec{n},kl} \cc{v_{\vec{n},k}}  v_{\vec{n},l} \ge b_1 \sum_{k=1}^d \abs{v_{\vec{n},k}}^2 \ .
\]
Since $b_1 > 0$, equation~\eqref{eq:zero2} may be true if and only if $v_{\vec{n},k} = 0$ for all $\vec{n} \in \mathbb{Z}^d$ and all $1 \le k \le d$. From~\eqref{eq:vdef}, we get $x_\vec{n} - x_{\vec{n} - \vec{e}_k} = 0$, which means that Theorem~\ref{thm:homconst} is true,
\[
x_\vec{n} = \const \ .
\]
\end{proof}

Using the above result, we prove the main theorem:
\begin{theorem}
\label{thm:uniq}
Let $f,f': \mathbb{Z}^d \mapsto \mathbb{C}$ be solutions of equation~\eqref{eq:Peq}, which fulfill condition~\eqref{eq:condition}. Then
\[
f_\vec{n} = f'_\vec{n} + \const \ ,
\]
which means that solutions of~\eqref{eq:Peq} which fulfill condition~\eqref{eq:condition} are unique up to a constant.
\end{theorem}

\begin{proof}
The difference $x \defined f - f'$ is a solution of equation~\eqref{eq:Peq1} and fulfills condition~\eqref{eq:sumx2}. Therefore, Theorem~\ref{thm:homconst} applies and we have
\[
x_\vec{n} = \const \ .
\]
Thus,
\[
f_\vec{n} = f'_\vec{n} + \const \ .
\]
This ends the proof.
\end{proof}

\section{Necessary condition for the existence of solutions}

\begin{theorem}
If $f$, fulfilling condition~\eqref{eq:condition}, is the solution of~\eqref{eq:Peq}, then its right-hand side $g$ must be square-summable.
\end{theorem}
\begin{proof}
Indeed, from~\eqref{eq:Peq} we have
\begin{equation}
\label{eq:normg}
\begin{split}
\sum_{\vec{n} \in \mathbb{Z}^d} \abs{g_\vec{n}}^2 &= \sum_{\vec{n} \in \mathbb{Z}^d} \sum_{k,l=1}^d \sum_{k',l'=1}^d \cc{b_{\vec{n},kl}} b_{\vec{n},k'l'} \cc{(f_\vec{n} - f_{\vec{n} - \vec{e}_l} )} (f_\vec{n} - f_{\vec{n} - \vec{e}_{l'}}) \\
&\quad - \sum_{\vec{n} \in \mathbb{Z}^d} \sum_{k,l=1}^d \sum_{k',l'=1}^d \cc{b_{\vec{n},kl}} b_{\vec{n} + \vec{e}_{k'},k'l'} \cc{(f_\vec{n} - f_{\vec{n} - \vec{e}_l} )} ( f_{\vec{n} + \vec{e}_{k'}} - f_{\vec{n} + \vec{e}_{k'} - \vec{e}_{l'}} ) \\
&\quad + \sum_{\vec{n} \in \mathbb{Z}^d} \sum_{k,l=1}^d \sum_{k',l'=1}^d \cc{b_{\vec{n} + \vec{e}_k,kl}} b_{\vec{n} + \vec{e}_{k'},k'l'} \cc{( f_{\vec{n} + \vec{e}_{k}} - f_{\vec{n} + \vec{e}_{k} - \vec{e}_{l}} )} \\
&\quad\quad\times ( f_{\vec{n} + \vec{e}_{k'}} - f_{\vec{n} + \vec{e}_{k'} - \vec{e}_{l'}} ) \\
&\quad - \sum_{\vec{n} \in \mathbb{Z}^d} \sum_{k,l=1}^d \sum_{k',l'=1}^d \cc{b_{\vec{n} + \vec{e}_k,kl}} b_{\vec{n},k'l'} \cc{( f_{\vec{n} + \vec{e}_{k}} - f_{\vec{n} + \vec{e}_{k} - \vec{e}_{l}} )} (f_\vec{n} - f_{\vec{n} - \vec{e}_{l'}}) \\
\end{split}
\end{equation}
Each of these terms is bounded, for example:
\begin{multline}
\abs{\sum_{k,l=1}^d \sum_{k',l'=1}^d \cc{b_{\vec{n},kl}} b_{\vec{n},k'l'} \cc{(f_\vec{n} - f_{\vec{n} - \vec{e}_l} )} (f_\vec{n} - f_{\vec{n} - \vec{e}_{l'}})} \le\\\le
\sum_{k,l=1}^d \sum_{k',l'=1}^d \abs{b_{\vec{n},kl}} \abs{b_{\vec{n},k'l'}} \abs{f_\vec{n} - f_{\vec{n} - \vec{e}_l}} \abs{f_\vec{n} - f_{\vec{n} - \vec{e}_{l'}}} \le \\ \le
b_2^2 \sum_{k,l=1}^d \sum_{k',l'=1}^d \abs{f_\vec{n} - f_{\vec{n} - \vec{e}_l}} \abs{f_\vec{n} - f_{\vec{n} - \vec{e}_{l'}}} \ ,
\end{multline}
since for Hermitian $b_\vec{n}$, we have $\abs{b_{\vec{n},kl}} \le \norm{b_\vec{n}}$. For each product of the type $\abs{f_\vec{n} - f_{\vec{n} - \vec{e}_l}} \abs{f_\vec{n} - f_{\vec{n} - \vec{e}_{l'}}}$, we have
\begin{equation}
\sum_{\vec{n} \in \mathbb{Z}^d} \abs{f_\vec{n} - f_{\vec{n} - \vec{e}_l}} \abs{f_\vec{n} - f_{\vec{n} - \vec{e}_{l'}}} \le 
\sum_{k=1}^d \sum_{\vec{n} \in \mathbb{Z}^d} \abs{f_\vec{n} - f_{\vec{n} - \vec{e}_k}}^2 < \infty
\end{equation}
and analogously for other terms in~\eqref{eq:normg}. Using these results, we obtain
\begin{equation}
\label{eq:normg2}
\sum_{\vec{n} \in \mathbb{Z}^d} \abs{g_\vec{n}}^2 \le b_2^2 d^4 \sum_{k=1}^d \sum_{\vec{n} \in \mathbb{Z}^d} \abs{f_\vec{n} - f_{\vec{n} - \vec{e}_k}}^2 < \infty \ .
\end{equation}
Hence, $g$ is square-summable.
\end{proof}

\section{Summary}

We introduced the generalized discrete Poisson equation in $d$ dimensions. With the use of the Fourier transform, we proved the uniqueness up to a constant of the solutions with square-summable discrete derivatives. We also provide the necessary condition for the existence of solution. Because of the ubiquity of discrete equations and the Poisson equation in particular, this result is important for many areas of physics and mathematics.

\section{Acknowledgments}

The author expresses gratitude to prof. Jerzy Kijowski for inspiring this work and useful comments about the paper and to dr Andrzej Wakulicz for fruitful discussion.

\bibliography{uniqsol}

\providecommand{\bysame}{\leavevmode\hbox to3em{\hrulefill}\thinspace}
\providecommand{\MR}{\relax\ifhmode\unskip\space\fi MR }
\providecommand{\MRhref}[2]{%
  \href{http://www.ams.org/mathscinet-getitem?mr=#1}{#2}
}
\providecommand{\href}[2]{#2}
\begin{thebibliography}{1}

\bibitem{hankerson}
Darrel Hankerson, \emph{An existence and uniqueness theorem for difference
  equations}, SIAM J. Math. Anal. \textbf{20} (1989), 1208--1217.

\bibitem{kijowski3}
Jerzy Kijowski, Gerd Rudolph, and Cezary \'{S}liwa, \emph{Charge superselection
  sectors for scalar {QED} on the lattice}, Ann. Henri Poincare \textbf{4}
  (2003), no.~6, 1137--1167.

\bibitem{kijowski2}
Jerzy Kijowski, Gerd Rudolph, and Artur Thielmann, \emph{Algebra of observables
  and charge superselection sectors for {QED} on the lattice}, Commun. Math.
  Phys. \textbf{188} (1997), no.~3, 535--564.

\bibitem{kijowski1}
Jerzy Kijowski and Artur Thielmann, \emph{Quantum electrodynamics on a
  space-time lattice}, J. Geom. Phys. \textbf{19} (1996), 173--205.

\bibitem{veit}
Jan Veit, \emph{Boundary value problems for partial difference equations},
  Multidim. Syst. Sign. Proc. \textbf{7} (1996), 113--134.

\end{thebibliography}
\bibliographystyle{amsplain}

\end{document}